\documentclass{article}

\title{Topological containment of the 5-clique minus an edge in 4-connected graphs\footnote{This research was funded in part by Australian Research Council Discovery Project DP130100300.}}
\author{{Rebecca Robinson and Graham Farr}\\
	[1ex] Clayton School of Information Technology\\
	Monash University\\
	Clayton, Victoria, 3800\\
	Australia\\[1ex]
	Rebecca.Robinson@monash.edu, Graham.Farr@monash.edu
	}

\date{May 3, 2017}

\usepackage[dvips]{graphicx}
\usepackage{amsmath, amsthm, amssymb}
\usepackage{mathtools}
\usepackage{url}
\usepackage{float}

\newtheorem{thm}{Theorem}

\newtheorem{claim}{Claim}

\begin{document}
\maketitle

\begin{abstract}
The topological containment problem is known to be polynomial-time solvable for any fixed pattern graph $H$, but good characterisations have been found for only a handful of non-trivial pattern graphs. The complete graph on five vertices, $K_5$, is one pattern graph for which a characterisation has not been found. The discovery of such a characterisation would be of particular interest, due to the Haj\'{o}s Conjecture. One step towards this may be to find a good characterisation of graphs that do not topologically contain the simpler pattern graph $K_5^-$, obtained by removing a single edge from $K_5$.

This paper makes progress towards achieving this, by showing that every 4-connected graph must contain a $K_5^-$-subdivision.
\end{abstract}

\section{Introduction}

The topological containment relation for graphs goes back to Kuratowski's characterisation of planar graphs as those that do not contain subdivisions of either $K_5$ or $K_{3,3}$ \cite{Kuratowski30}. Since then, topological containment has been used to obtain further characterisations, such as for outerplanar graphs \cite{Chartrand67} and series-parallel graphs \cite{Dirac52, Duffin65}. It is known, thanks to Robertson and Seymour \cite{R&S85, R&S95}, that the question of whether a graph contains an $H$-subdivision is always solvable in polynomial time for fixed $H$. More recently, Grohe and Marx \cite{Grohe11, Grohe11_2} established that the problem is fixed-parameter tractable, with parameter $H$. Nonetheless, good characterisations have been found for only a few non-trivial pattern graphs $H$: $K_{1,n}$, $K_3$ (elementary), $C_4$, $C_5$ \cite{SeymourWalton81}, $C_6$, $C_7$ \cite{FarrThesis86}, $K_4$ \cite{Dirac52, Duffin65}, $K_{3,3}$ \cite{Wagner37, Hall43}, $W_4$, $W_5$ \cite{Farr88}, $W_6$ \cite{Robinson08}, and $W_7$ \cite{Robinson14} (where $W_n$ denotes the wheel with $n$ spokes), and more recently, the Wagner graph, also known as $V_8$ \cite{MaharryRobertson16}.

One pattern graph of particular interest in the context of this problem is $K_5$, due to its connections with the Haj\'{o}s Conjecture. It was conjectured by Haj\'{o}s in the 1940s that any graph with no $K_k$-subdivision is $(k-1)$-colourable. This is known for $k \le 4$ \cite{Hadwiger43, Dirac52}, and has been refuted for $k \ge 7$ \cite{Catlin79}, but for $k = 5$ and $k = 6$ the conjecture remains open. A good characterisation for graphs containing no $K_5$-subdivision may help solve this problem for the $k = 5$ case. While such a characterisation has not been found, a handful of results have been found that lead to further understanding of the structure of graphs with no $K_5$-subdivision. Mader \cite{Mader98-2} showed that every simple graph on $n \ge 5$ vertices and with at least $3n - 5$ edges contains a $K_5$-subdivision, thus proving an early conjecture of Dirac \cite{Dirac52}. Seymour (unpublished), and separately, Kelmans \cite{Kelmans79}, conjectured that every 5-connected nonplanar graph contains a $K_5$-subdivision. This conjecture has been established for graphs containing $K_4$ minus an edge \cite{MaYu2010, MaYu2013}, and for graphs containing $K_{2,3}$ in \cite{Kawara2015}. A proof of the Kelmans-Seymour Conjecture for all 5-connected nonplanar graphs has recently been announced in \cite{HeWangYu2015, HeWangYu2016, HeWangYu2016-III, HeWangYu2016-IV}.

A step along the way to characterising graphs with no $K_5$-subdivision is to solve the problem for a slightly simpler graph, $K_5^{-}$, which we obtain by removing a single edge from $K_5$. In this paper, we take a step in this direction, showing that every 4-connected graph contains a $K_5^{-}$-subdivision. The approach used is similar to that in \cite{Robinson08, Robinson14}, where firstly a `base' graph is identified for which a good characterisation of topological containment is already known, and which is a subgraph of the pattern graph $H$ (here, $H = K_5^{-}$), then we look at all possible ways of enlarging this base graph so that the conditions of the hypothesis are met (in this case, 4-connectivity). For each enlarged graph, we determine whether this graph contains an $H$-subdivision. Here, the base graph chosen is $W_4$, firstly since we know as a consequence of the characterisation in \cite{Farr88} that any 4-connected graph must contain a $W_4$-subdivision, and secondly since $W_4$ differs from $K_{5}^{-}$ only by a single edge.

A complete characterisation for $K_5^{-}$ could potentially form the basis for a similarly structured proof characterising graphs with no $K_5$-subdivision.

This approach necessarily involves checking many specific graphs for the presence of a subdivision of the pattern graph $H$.  These graphs are not large, so that each check is easy to do.  We have found that the number of these checks grows rapidly as the pattern graph increases in size.  For some small $H$, the checking can be done by hand, as for $W_4$ and $W_5$ \cite{Farr88} and many others in the above list.  For others, computer assistance is needed, as for $W_6$ and $W_7$ in \cite{Robinson08, Robinson14}.  In this paper, with pattern graph $K_5^{-}$, we find that the amount of checking, although nontrivial, is within reach of manual verification, so that computer assistance is not necessary.

The result given in this paper can also be considered as a step in parallel to the Kelmans-Seymour Conjecture, giving insight into the topological structure of 4-connected graphs in much the same way as the Kelmans-Seymour Conjecture does for 5-connected graphs. In fact, a family of related results can be observed, for graphs of increasing connectivity:

\begin{itemize}
\item 2-connectivity implies topological containment of $K_3$;
\item 3-connectivity implies topological containment of $K_4$;
\item 3-connectivity with some vertex of degree $\ge 4$ implies topological containment of $W_4$ \cite{Farr88};
\item 4-connectivity implies topological containment of $K_5^{-}$ (this paper);
\item 5-connectivity implies topological containment of $K_5$ (Kelmans-Seymour Conjecture).
\end{itemize}

An interesting result relating to minor containment is due to Halin and Jung \cite{HalinJung63}, who show that every 4-connected graph contains $K_5$ or $K_{2,2,2}$ as a minor. Note that subgraphs contractible to $K_5$ or $K_{2,2,2}$ do not necessarily contain a $K_5^-$-subdivision.

\section{Some definitions}
The \emph{neighbourhood} $N_{G}(v)$ of a vertex $v$ in $G$ is the set of vertices which are adjacent to $v$ in $G$.

Given a path $P$ where $x, y \in V(P)$, we denote by $xPy$ the subpath of $P$ between $x$ and $y$, including $x$ and $y$. A \emph{proper subpath} of $P$ is a subpath of $P$ other than $P$ itself.

If $H$ is a subgraph of some graph $G$, an \emph{$H$-bridge} in $G$ is a subgraph of $G$ which is either an edge not in $H$ but with both ends in $H$ (an \emph{inner} $H$-bridge), or a connected component of $G - V(H)$ together with all edges (including their endvertices) that join this component to $H$ (an \emph{outer} $H$-bridge). This definition is from \cite{Mohar01, Tutte66}.

Let $U$ be an $H$-bridge of $G$. A \emph{vertex of attachment} is a vertex in $V(U)\cap V(H)$. An edge of $U$ incident with a vertex of attachment is a \emph{foot} of $U$. This terminology is also used in \cite{Mohar01, Tutte66}.

\vspace{0.1in}
\noindent \textbf{Wheel terminology}

\begin{itemize}
\item The \emph{hub} of a $W_n$-subdivision is the vertex of degree $n$ in that wheel subdivision.
\item The \emph{rim} of a  $W_{n}$-subdivision is the cycle around the outside of that wheel subdivision (excluding the hub).
\item The \emph{spoke-meets-rim vertices} of a $W_{n}$-subdivision are the $n$ vertices of degree 3 in that wheel subdivision.
\item The \emph{spokes} of a  $W_{n}$-subdivision are the $n$ paths from the hub to the spoke-meets-rim vertices in that wheel subdivision. They each meet the rim only at one spoke-meets-rim vertex, and are vertex-disjoint except at the hub.
\item The \emph{segments} of the rim in a  $W_{n}$-subdivision are the $n$ paths that form subpaths of the rim, such that each segment has two spoke-meets-rim vertices as endpoints, and does not contain any spoke-meets-rim vertices internally.
\item Two spokes of a wheel subdivision are said to be \emph{neighbouring} spokes if their spoke-meets-rim vertices have only a single rim segment between them.
\end{itemize}

Let $P_i$ be some spoke of a  $W_{n}$-subdivision $H$. An \emph{initial segment} of $P_i$ is some subpath of $P_i$ that has as one of its endpoints the hub of $H$. A \emph{proper initial segment} of $P_i$ is an initial segment of $P_i$ that is also a proper subpath of $P_i$.

Let $H$ be a $W_n$-subdivision in a graph $G$. We say that another $W_n$-subdivision $J$ in $G$ is \emph{shorter than} $H$ if:

\begin{itemize}
\item the hubs of $H$ and $J$ are the same;
\item the spokes of $H$ and $J$ are not all the same;
\item each spoke of $J$ is an initial segment of a spoke of $H$ (that is, for each spoke $P_i$ of $H$, there exists a vertex $w_i$ on $P_i$ such that $vP_iw_i$ is a spoke of $J$, where $v$ is the hub of both $H$ and $J$); and
\item at least one spoke of $J$ is a \emph{proper} initial segment of a spoke of $H$.
\end{itemize}

If no other $W_n$-subdivision in $G$ is shorter than $H$, then we say that $H$ is \emph{short}.

A short $W_n$-subdivision $H$ need not necessarily have minimum sum of spoke lengths, over all $W_n$-subdivisions in $G$.  But any $W_n$-subdivision with lower total spoke length than $H$ must have some spoke that is not an initial segment of a spoke of $H$.  Also, any $W_n$-subdivision of minimum total spoke length must be short.

It is clear that, if $G$ has a $W_n$-subdivision, then it has a short $W_n$-subdivision.

\vspace{0.1in}
\noindent \textbf{$K_{5}^{-}$-subdivision terminology}
\begin{itemize}
\item The \emph{trivertices} of a $K_{5}^{-}$-subdivision are the vertices of degree 3 in that subdivision.
\item The \emph{tetravertices} of a $K_{5}^{-}$-subdivision are the vertices of degree 4 in that subdivision.
\item The \emph{terminal vertices} of a $K_{5}^{-}$-subdivision $H$ are all vertices that are trivertices or tetravertices of $H$.
\end{itemize}

\section{Result}

\begin{thm}
\label{4conn}
Let $G$ be a 4-connected graph. $G$ contains a $K_{5}^{-}$-subdivision.
\end{thm}

\begin{proof}
By Theorem 1 and Lemma 2 in \cite{Farr88}, there exists a $W_{4}$-subdivision in $G$. Let $H$ be a $W_4$-subdivision in $G$, chosen such that $H$ is \emph{short}.

Let $v$ be the hub of $H$, and let $v_1, \ldots, v_4$ be the spoke-meets-rim vertices in order around the rim of $H$. Let $P_1, \ldots, P_4$ be the spokes of $H$ with endpoints $v_1, \ldots, v_4$ respectively. Let $R_1, \ldots, R_4$ be the rim segments in order around the rim of $H$, starting at $v_1$, with $R_i$ going from $v_i$ to $v_{i+1}$, $i\le 3$, and $R_4$ going from $v_4$ to $v_1$.

By the 4-connectivity of $G$, there must exist some fourth neighbour $u_1$ of $v_1$, such that $u_1 \notin N_{H}(v_1)$. Also to preserve 4-connectivity, there must be at least three paths from $u_1$ to $H - v_1$, disjoint except at $u_1$, that meet $H$ only at their endpoints in $H - v_1$. Let one such path be called $P'$, and let $P$ be the path defined by $P' + v_1u_1$.

Let $p_1$ be the vertex at which $P$ meets $H$.  Consider the possibilities for this, which we group into five cases, (a) -- (e):

\begin{itemize}
\item[(a)] $p_1$ is an internal vertex of $P_2$ or $P_4$;
\item[(b)] $p_1 = v_3$;
\item[(c)] $p_1$ is an internal vertex of $R_2$ or $R_3$;
\item[(d)] $p_1$ is an internal vertex of $P_3$;
\item[(e)] $p_1$ lies on $R_1$, $R_4$, or $P_1$.
\end{itemize}

We treat each of these cases in turn. The first two are easily dealt with. The third, (c), is the most complex.


\vspace{0.1in}
\noindent \textbf{Case (a): $p_1$ is an internal vertex of $P_2$ or $P_4$.}

Assume (without loss of generality, due to symmetry) that $p_1$ is an internal vertex of $P_2$. Then there exists a new $W_4$-subdivision, $H'$, centred on $v$, with spoke-meets-rim vertices $v_1$, $v_4$, $v_3$, and $p_1$. This contradicts our definition of $H$ as short.


\vspace{0.1in}
\noindent \textbf{Case (b): $p_1 = v_3$.}

If $p_1 = v_3$, then a $K_{5}^{-}$-subdivision exists in $G$, with trivertices $v_2$ and $v_4$ and tetravertices $v$, $v_1$, and $v_3$.


\vspace{0.1in}
\noindent \textbf{Case (c): $p_1$ is an internal vertex of $R_2$ or $R_3$.}

Assume (without loss of generality, due to symmetry) that $p_1$ is an internal vertex of $R_2$. Assume also that $P$ is chosen to minimise the distance between $p_1$ and $v_3$ along $R_2$.

By the 4-connectivity of $G$, there must exist some fourth neighbour $u_3$ of $v_3$, such that $u_3 \notin N_{H}(v_3)$. Let $U_3$ be the $(H\cup P)$-bridge of $G$ containing the edge $v_3u_3$.

If $v_1$ is a vertex of attachment of $U_3$, then there exists a path from $v_1$ to $v_3$ such that the graph falls into Case (b), and a $K_{5}^{-}$-subdivision is easily formed. Assume then that $v_1$ is not a vertex of attachment of $U_3$.

If any of $U_3$'s vertices of attachment lie on $R_1 - v_1$, on $v_2R_2p_1 - p_1$, internally on $P$, or internally on $P_1$, then a $K_{5}^{-}$-subdivision can be formed, as shown in the four graphs of Figure \ref{case-c-1}. Note that the graphs illustrated represent graphs that are contained in $G$ as a subdivision; as such, a single edge in these graphs represents a path in $G$.

Assume then that $U_3$ does not have vertices of attachment on $R_1 - v_1$, $v_2R_2p_1 - p_1$, internally on $P$, or internally on $P_1$.

\begin{figure}[H]
\begin{center}
\includegraphics[width=0.8\textwidth]{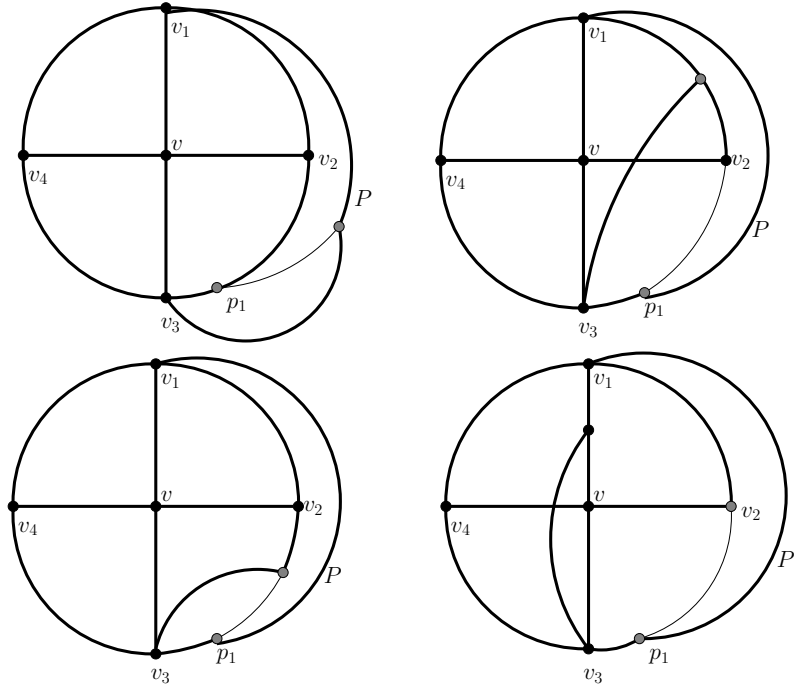}
\caption[Case (c): $K_5^{-}$-subdivisions created by paths in $U_3$.]{Case (c). A $K_5^{-}$-subdivision exists if $U_3$ has vertices of attachment internally on $P$, on on $R_1 - v_1$, on $v_2R_2p_1 - p_1$, or internally on $P_1$. In each case, the edges of the $K_5^{-}$-subdivision are shown in bold, and non-terminal vertices are grey.}
\label{case-c-1}
\end{center}
\end{figure}

If any of $U_3$'s vertices of attachment are internal vertices of $P_2$ or $P_4$, then there exists some $W_4$-subdivision centred on $v$ such that $H$ is no longer short, which contradicts our choice of $H$.

It can be assumed, then, that $U_3$'s vertices of attachment lie only on the following paths:

\begin{itemize}
\item[(i)] $R_4$ (internally)
\item[(ii)] $p_1R_2v_3$, $R_3$, or $P_3$ (potentially at their endpoints)
\end{itemize}

These cases are considered below.

\vspace{0.1in}
\noindent \textbf{(i) $U_3$ has a vertex of attachment internally on $R_4$}

Suppose $U_3$ has a vertex of attachment that lies internally on $R_4$. Call this vertex $q_3$. Then there exists some path from $q_3$ to $v_3$ that is internally disjoint from $H\cup P$; call this path $Q$. See Figure \ref{case-ci}.

\begin{figure}[H]
\begin{center}
\includegraphics[width=0.5\textwidth]{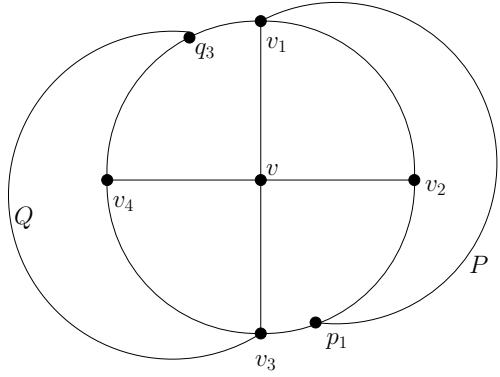}
\caption[Case (c)(i): Paths $P$ and $Q$.]{Case (c)(i). Path $Q$ from $v_3$ to $q_3$, where $q_3$ lies internally on $R_4$.}
\label{case-ci}
\end{center}
\end{figure}

Without loss of generality, assume that of all paths from $v_3$ to an internal vertex of $R_4$ that are internally disjoint from $H\cup P$, $Q$ is chosen such that the distance between $q_3$ and $v_1$ along $R_4$ is minimised.

Let $G_1 = R_1\cup R_2\cup P_1\cup P_2\cup P$. Let $G_2 = R_3\cup R_4\cup P_3\cup P_4\cup Q$. Observe that $V(G_1)\cap V(G_2) = \{v_1, v, v_3\}$. By 4-connectivity, there must exist some path $R$ from $G_1$ to $G_2$ that is disjoint from $\{v_1, v, v_3\}$ and meets $G_1\cup G_2$ only at its endpoints.

The endpoints of $R$ can be as follows:

\begin{tabular}{l l l}
1. $R_1$ to $R_3$ & 2a. $R_1$ to $v_4R_4q_3$ & 2b. $R_1$ to $q_3R_4v_1$ \\
3. $R_1$ to $P_3$ & 4. $R_1$ to $P_4$ & 5. $R_1$ to $Q$ \\
6a. $v_2R_2p_1$ to $R_3$ & 6b. $p_1R_2v_3$ to $R_3$ & 7a. $v_2R_2p_1$ to $v_4R_4q_3$ \\
7b. $p_1R_2v_3$ to $v_4R_4q_3$ & 7c. $v_2R_2p_1$ to $q_3R_4v_1$ & 7d. $p_1R_2v_3$ to $q_3R_4v_1$ \\
8a. $v_2R_2p_1$ to $P_3$ & 8b. $p_1R_2v_3$ to $P_3$ & 9a. $v_2R_2p_1$ to $P_4$ \\
9b. $p_1R_2v_3$ to $P_4$ & 10a. $v_2R_2p_1$ to $Q$ & 10b. $p_1R_2v_3$ to $Q$ \\
11. $P_1$ to $R_3$ & 12. $P_1$ to $R_4$ & 13. $P_1$ to $P_3$ \\
14. $P_1$ to $P_4$ & 15. $P_1$ to $Q$ & 16. $P_2$ to $R_3$ \\
17a. $P_2$ to $v_4R_4q_3$ & 17b. $P_2$ to $q_3R_4v_1$ & 18. $P_2$ to $P_3$ \\
19. $P_2$ to $P_4$ & 20. $P_2$ to $Q$ & 21. $P$ to $R_3$ \\
22a. $P$ to $v_4R_4q_3$ & 22b. $P$ to $q_3R_4v_1$ & 23. $P$ to $P_3$ \\
24. $P$ to $P_4$ & 25. $P$ to $Q$ \\
\end{tabular}

Table \ref{case-ci-table} shows how each case number listed above is assigned to each possible pairing of endpoints.

\begin{table}[H]
\begin{center}
\begin{tabular}{|p{1.2cm}||c|c|c|c|c|c|} \hline
\hfill $G_2$\newline $G_1$ & $R_3$ & $v_4R_4q_3$ & $q_3R_4v_1$ & $P_3$ & $P_4$ & $Q$ \\ \hline\hline
$R_1$ & 1 & 2a & 2b & 3 & 4 & 5 \\ \hline
$v_2R_2p_1$ & 6a & 7a & 7c & 8a & 9a & 10a \\ \hline
$p_1R_2v_3$ & 6b & 7b & 7d & 8b & 9b & 10b \\ \hline
$P_1$ & 11 & 12 & 12 & 13 & 14 & 15 \\ \hline
$P_2$ & 16 & 17a & 17b & 18 & 19 & 20 \\ \hline
$P$ & 21 & 22a & 22b & 23 & 24 & 25 \\ \hline
\end{tabular}
\caption[Case (c)(i). Table showing possible endpoints of $R$.]{Case (c)(i). Table showing assigned case numbers of all possible pairs of endpoints of $R$.}
\label{case-ci-table}
\end{center}
\end{table}

In cases 1, 2a, 4, 5, 6a, 7a, 10a, 16, 21, 22a, and 25, a $K_5^-$-subdivision exists. In cases 3, 8a, 8b, 9a, 9b, 11, 12, 14, 15, 17a, 17b, 18, 19, 20, 23, and 24, a $W_4$-subdivision centred on $v$ can be found such that $H$ is no longer short. This leaves only the following cases where $R$ either goes from $P_1$ to $P_3$ (case 13), or has an endpoint either in $p_1R_2v_3$ or $q_3R_4v_1$:

\begin{itemize}
\item[2b.] $R_1$ to $q_3R_4v_1$
	\item[6b.] $p_1R_2v_3$ to $R_3$
	\item[7b.] $p_1R_2v_3$ to $v_4R_4q_3$
\item[7c.] $v_2R_2p_1$ to $q_3R_4v_1$
	\item[7d.] $p_1R_2v_3$ to $q_3R_4v_1$
	\item[10b.] $p_1R_2v_3$ to $Q$
		\item[13.] $P_1$ to $P_3$
\item[22b.] $P$ to $q_3R_4v_1$
\end{itemize}

Assume then that \emph{any} path from $G_1$ to $G_2$ meeting the criteria for $R$ (disjoint from $\{v_1, v, v_3\}$, and meets $G_1\cup G_2$ only at its endpoints) falls into one of these eight cases.

\textbf{Suppose firstly that case 13 holds,} that is, there exists such a path $R$ from $P_1$ to $P_3$. Now, let $G'_1 = G_1 - P_1$, and let $G'_2 = G_2 - P_3$. By 4-connectivity, there must exist some path $R'$ from $G'_1$ to $G'_2$ that is disjoint from $\{v_1, v, v_3\}$ and meets $G'_1\cup G'_2$ only at its endpoints. (It may or may not meet $R$.) If this path meets $P_1$ or $P_3$, then either it falls into one of the cases already dealt with above (one of cases 3, 8a, 8b, 11, 12, 14, 15, 18, or 23), or $R'$ has a subpath, internally disjoint from $G_1\cup G_2$, with one of the following pairs of endpoints:

\begin{tabular}{l l}
(i) $P_1$ and $R_1$ & (ii) $P_1$ and $R_2$ \\
(iii) $P_1$ and $P_2$ & (iv) $P_1$ and $P$ \\
(v) $P_3$ and $R_3$ & (vi) $P_3$ and $R_4$ \\
(vii) $P_3$ and $P_4$ & (viii) $P_3$ and $Q$ \\
\end{tabular}

In each of these eight cases, a $W_4$-subdivision centred on $v$ can be created such that $H$ is no longer short (regardless of whether or not the subpath meets the path $R$ from $P_1$ to $P_3$). Figure \ref{case-ci-13} shows two examples of such graphs, from endpoint pairs (i) and (ii) respectively, with the subpath of $R'$ shown in blue, and the shorter-spoked $W_4$-subdivision in bold.

\begin{figure}[H]
\begin{center}
\includegraphics[width=\textwidth]{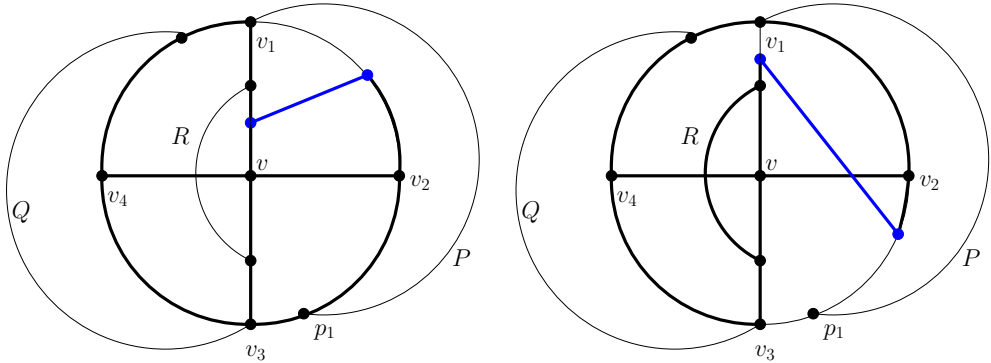}
\caption[Case (c)(i), 13. $W_4$-subdivisions with shorter spokes than $H$.]{Case (c)(i); case 13 of path $R$. Examples of how a subpath of $R'$ forms a $W_4$-subdivision with shorter spokes than $H$.}
\label{case-ci-13}
\end{center}
\end{figure}

Assume, then, that any such path $R'$ from $G'_1$ to $G'_2$ does not meet $P_1$ or $P_3$. The set of all possible paths $R'$ is a subset of those paths meeting the criteria for $R$ (excluding those that have $P_1$ or $P_3$ as an endpoint). We have already shown that many of the paths in this subset result either in a $K_5^-$-subdivision, or in a violation of the shortness of $H$. We assume, then, that the endpoints of $R'$ are among those remaining possibilities where this is not the case:

\begin{itemize}
\item[2b.] $R_1$ to $q_3R_4v_1$
	\item[6b.] $p_1R_2v_3$ to $R_3$
	\item[7b.] $p_1R_2v_3$ to $v_4R_4q_3$
\item[7c.] $v_2R_2p_1$ to $q_3R_4v_1$
	\item[7d.] $p_1R_2v_3$ to $q_3R_4v_1$
	\item[10b.] $p_1R_2v_3$ to $Q$
\item[22b.] $P$ to $q_3R_4v_1$
\end{itemize}

These possibilites for the placement of $R'$ are the same as the remaining cases we have left to look at for path $R$ (which currently runs from $P_1$ to $P_3$). So we have shown that if case 13 holds, then we can construct another path (namely $R'$) which can play the role of $R$ and puts us into one of the other cases in the above list.

Thus, we will \textbf{assume one of cases 2b, 6b, 7b, 7c, 7d, 10b, 22b holds}, and return to considering these remaining cases for $R$ --- but noting that some other path meeting the criteria for $R$ with endpoints on $P_1$ and $P_3$ may or may not also exist.

Without loss of generality (due to symmetry of the graph), suppose that one of cases 6b, 7b, 7d, or 10b hold. (Cases 2b, 7c, and 22b are symmetrically equivalent to 6b, 7b, and 10b respectively.) In each of these cases, $R$ has an endpoint in $p_1R_2v_3$. Call this endpoint $r_1$. Let $r_1$ be chosen to minimise the distance along $R_2$ between $r_1$ and $p_1$.

If there exists, in addition to $R$, another path that falls into one of the `symmetrically equivalent' cases 2b, 7c, or 22b, then a $K_5^-$-subdivision exists. (There are twelve possible graphs to consider---with four possible placements of $R$, and three possible placements of the additional path.) Assume then that such a path does not exist. If in addition to $R$, there exists a path that falls into case 13 (i.e., a path from $P_1$ to $P_3$), then a $W_4$-subdivision centred on $v$ exists such that $H$ is no longer short. Figure \ref{case-ci-13-2} shows the four possible graphs, with the four possible placements of $R$ in blue, and the new $W_4$-subdivision in bold.

\begin{figure}[H]
\begin{center}
\includegraphics[width=\textwidth]{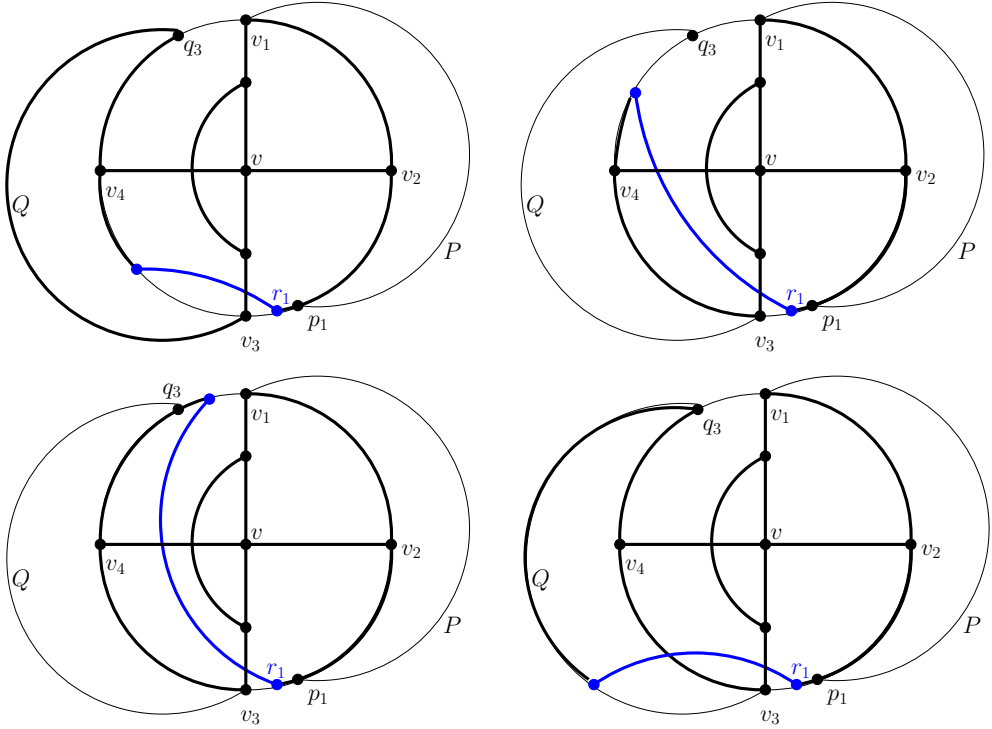}
\caption[Case (c)(i). Path $R$ with endpoint on $p_1R_2v_3$; $W_4$-subdivisions with shorter spokes than $H$.]{Case (c)(i). Path $R$ with endpoint on $p_1R_2v_3$; additional path from $P_1$ to $P_3$ creates a $W_4$-subdivision with shorter spokes than $H$.}
\label{case-ci-13-2}
\end{center}
\end{figure}

We have now dealt with the possibility of other paths (as well as $R$) satisfying the same criteria as $R$ and falling within cases 2b, 7c, 13, and 22b.

Thus, we may assume that \emph{any} path from $G_1$ to $G_2$ that meets the criteria for $R$ falls within one of cases 6b, 7b, 7d, 10b, so it meets $G_1$ \emph{only} on the path $p_1R_2v_3$. Of all these paths, $R$ is chosen such that its endpoint $r_1$ in $G_1$ is closest to $p_1$ along $R_2$.

If there is no path from $r_1R_2v_3 - r_1$ to $G_1 - P_1 - r_1R_2v_3$ that is internally disjoint from $H\cup P\cup Q$, then the graph can be disconnected with the removal of $r_1$, $v$, and $v_1$, violating 4-connectivity. Assume then that such a path exists. Call this path $R'$.

$R'$ cannot meet $P$ internally, due to the choice of $P$ minimising the distance between $p_1$ and $v_3$ (as stated at the beginning of Case (c)). $R'$ cannot meet $P_2$ internally, or a new $W_4$-subdivision can be constructed such that $H$ is no longer short. Suppose all possible choices of $R'$ meet $H$ only on $r_1R_2p_1 - r_1$. Then we can construct a new segment of rim $R'_2$ (from parts of $R_2$ and some choice of path $R'$) and a new path $R$ (from the old path $R$, and parts of $R_2$), such that the graph can be disconnected as in the previous paragraph. Thus, there must exist some $R'$ such that its endpoint in $G_1 - P_1 - r_1R_2v_3$ is on $R_1$ or $v_2R_2p_1$.

There are eight possible graphs to consider (with four possible placements of $R$, and two of $R'$), and each resulting graph contains a $K_5^-$-subdivision, \emph{except} for the two cases where $R$ meets $G_2$ on the path $v_1R_4q_3$. Assume then that this is the case for \emph{all} paths joining $G_1$ to $G_2$ that are disjoint from $\{v_1, v, v_3\}$. Let $r_2$ be a vertex on $v_1R_4q_3$ that forms the endpoint of some such path $R$. By a similar argument to the previous paragraph, there must be some path joining $r_2R_4v_1$ to either $R_3$ or $v_4R_4q_3$. Such a path in conjunction with $R$ and $R'$ (regardless of the placement of each path --- this time there are four graphs to consider, with two possible placements of $R'$, and two of the new path) results in a $K_5^-$-subdivision.

\vspace{0.1in}
\noindent \textbf{(ii) $U_3$ only has vertices of attachment on $p_1R_2v_3$, $R_3$, or $P_3$}

Suppose $U_3$ only has vertices of attachment on $p_1R_2v_3$, $R_3$, or $P_3$.

Let $G_1 = p_1R_2v_3\cup R_3\cup P_3$. Let $G_2 = (H\cup P) - G_1$. Note that throughout this section of the proof, $U_3$'s vertices of attachment are all contained in $G_1$.

\textbf{1.} Suppose $U_3$ has some vertex of attachment $q_3$ on $P_3$ \emph{other than} $v_3$. Let $Q_3$ be a path from $v_3$ to $q_3$, contained in $U_3$, that does not meet $H\cup P$ internally.

Let $\mathcal{P}_{a,b}$ be the set of all pairs of internally disjoint paths $\{P_a, P_b\}$, such that:

\begin{itemize}
\item $P_a$, $P_b$ have shared endpoints, one of which is $v_3$, and the other some vertex on $P_3 - v_3$;
\item neither $P_a$ nor $P_b$ meet $(H\cup P)-P_3$ internally. (They may, however, interact with $P_3$.)
\end{itemize}

Note that $\{Q_3, v_3P_3q_3\}$ is one such pair of paths belonging to $\mathcal{P}_{a,b}$, so the set is non-empty.

Let $\{P_x, P_y\}$ be a pair of paths belonging to $\mathcal{P}_{a,b}$ with endpoints $v_3$ and $x$, chosen so that the distance between $v$ and $x$ along $P_3$ is minimised. See Figure \ref{case-cii-1}. ($G_1$ is shown in blue, and $G_2$ in red. Note that the path $v_3P_3x$, although part of $G_1$, is not shown here: it may or may not interact with $P_x$ and $P_y$ internally.)

\begin{figure}[H]
\begin{center}
\includegraphics[width=0.5\textwidth]{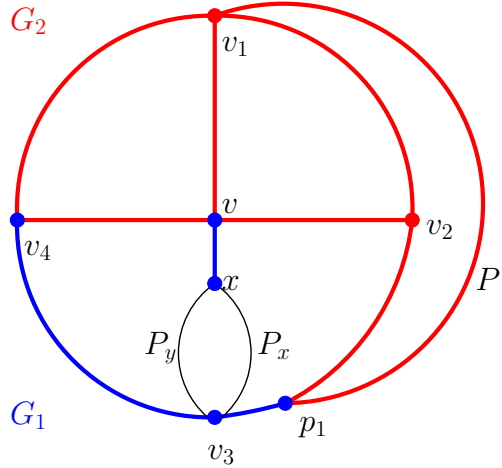}
\caption[Case (c)(ii)1. $G_1$, $G_2$, $P_x$, and $P_y$.]{Case (c)(ii)1. $G_1$, $G_2$, $P_x$, and $P_y$.}
\label{case-cii-1}
\end{center}
\end{figure}

\begin{claim}
\label{claim_cii1_R}
If there exists a path $R$ in $G$ from $(P_x\cup P_y) - \{v_3, x\}$ to somewhere in $G_2$ that is internally disjoint from $(H - v_3P_3x)\cup P\cup P_x\cup P_y$, then a $K_5^-$-subdivision exists in $G$.
\end{claim}

\begin{proof}
Suppose such a path $R$ exists. Let $r_1$ be the endpoint of $R$ that lies on $(P_x\cup P_y) - \{v_3, x\}$. Without loss of generality, suppose $r_1$ lies internally on $P_y$.

Let $P'_3$ be the path from $v_3$ to $v$ formed from $vP_3x$ and $P_x$. $P'_3$ does not meet any internal vertex of $P_y$.

Let $Q'$ be the path from $v_3$ to somewhere in $G_2$ formed from $v_3P_yr_1$ and $R$. This path does not meet $P'_3$ except at $v_3$. If $Q'$'s endpoint in $G_2$ lies on $P_1 - v$, $R_1$, $v_2R_2p_1 - p_1$, or $P - p_1$, then a $K_{5}^{-}$-subdivision can be formed, in the same way as described in the first few paragraphs of Case (c) and illustrated in Figure \ref{case-c-1}. If $Q'$'s endpoint in $G_2$ is on $P_2$ or $P_4$, then there exists a $W_4$-subdivision centred on $v$ (with $x$ and $r_2$ as two of its spoke-meets-rim vertices) such that $H$ is no longer short. Figure \ref{case-cii-R1} shows how this $W_4$-subdivision would be formed.

\begin{figure}[H]
\begin{center}
\includegraphics[width=\textwidth]{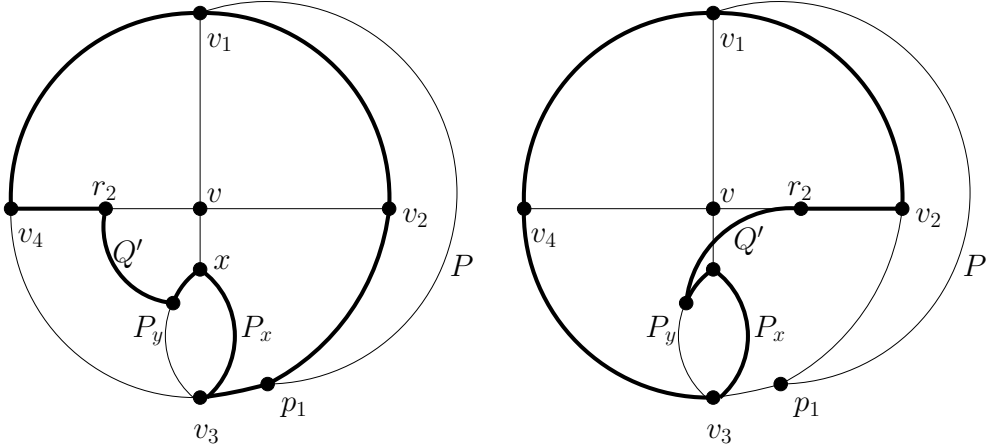}
\caption[Case (c)(ii)1, Claim \ref{claim_cii1_R}.]{Case (c)(ii)1, Claim \ref{claim_cii1_R}: $Q'$ goes to $P_2$ or $P_4$, violating shortness of $H$. The shorter $W_4$-subdivision has its rim shown in bold.}
\label{case-cii-R1}
\end{center}
\end{figure}

It remains to consider the possibility that $Q'$'s endpoint in $G_2$ is an internal vertex of $R_4$. See Figure \ref{case-cii-2}.

\begin{figure}[H]
\begin{center}
\includegraphics[width=0.5\textwidth]{case-cii-2.ps}
\caption[Case (c)(ii)1. Paths $P'_3$ and $Q'$ in Claim \ref{claim_cii1_R}.]{Case (c)(ii)1. Paths $P'_3$ and $Q'$ in Claim \ref{claim_cii1_R}.}
\label{case-cii-2}
\end{center}
\end{figure}

Let $H'$ be a $W_4$-subdivision that coincides with $H$ except on $P_3$, and has the spoke $P'_3$ instead of $P_3$.

We wish to appeal to Case (c)(i), with $W_4$-subdivision $H'$ and path $Q'$ playing the roles of $H$ and $Q$ respectively. But the argument used there relies on $H$ being short. In the present situation, we do not yet know that $H'$ is short.

Suppose $H'$ is not short. Then there exists another $W_4$-subdivision, $H''$, centred on $v$, whose spokes are initial segments of the spokes of $H'$, and at least one is proper. Note that the rim of $H''$ may encounter its spokes in a different cyclic order to that used by the rim of $H$.

Consider $P''_3$, the spoke of $H''$ that is an initial segment $vP'_3y$ of $P'_3$, where $y$ is the spoke-meets-rim vertex of $H''$ that lies on $P'_3 - v$. There must exist two internally-disjoint paths $R'_2$ and $R'_3$ that form two segments of the rim of $H''$, such that these paths go from $y$ to two distinct spokes of $H'$ other than $P'_3$, that is, two of $P_1$, $P_2$, and $P_4$. One of these paths, then --- assume $R'_2$, without loss of generality --- meets either $P_2$ or $P_4$.  Let $z$ be the point closest to $y$ along $R'_2$ (other than $y$) at which $R'_2$ first meets $H$, and let $S'_2$ be the subpath of $R'_2$ from $y$ to $z$. See Figure \ref{case-cii-3} for one possible configuration, where $R'_2$ does not meet $H$ internally, and so $S'_2 = R'_2$.

\begin{figure}[H]
\begin{center}
\includegraphics[width=0.5\textwidth]{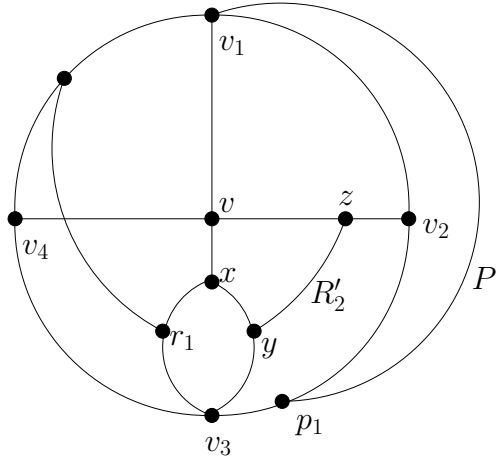}
\caption[Case (c)(ii)1. Claim \ref{claim_cii1_R}: $y$ and $R'_2$.]{Case (c)(ii)1. Claim \ref{claim_cii1_R}: a possible configuration of $y$ and $R'_2$.}
\label{case-cii-3}
\end{center}
\end{figure}

Recall $P'_3$ is composed of two subpaths, $vP_3x$ and $P_x$. Suppose $y$ lies on $vP_3x$ or internally on $P_x$.  Figure \ref{case-cii-3} shows one instance of this, and in this situation it is clear that a $W_4$-subdivision is constructed which violates the shortness of $H$. The rim of the forbidden $W_4$-subdivision coincides with that of $H$, except on $R_2$, where instead the rim is formed by the path $v_3P_yxP_xyR'_2zP_2v_2$.

However, the vertex $z$ may lie on a rim segment of $H$, if $S'_2 \neq R'_2$. If $z$ lies on $R_1$, then a $K_5^-$-subdivision exists. If $z$ lies on $R_2$, $R_3$, or $R_4$, then again, a $W_4$-subdivision is constructed which violates the shortness of $H$. Figure \ref{case-cii-3a} illustrates each of these possibilities.

\begin{figure}[H]
\begin{center}
\includegraphics[width=\textwidth]{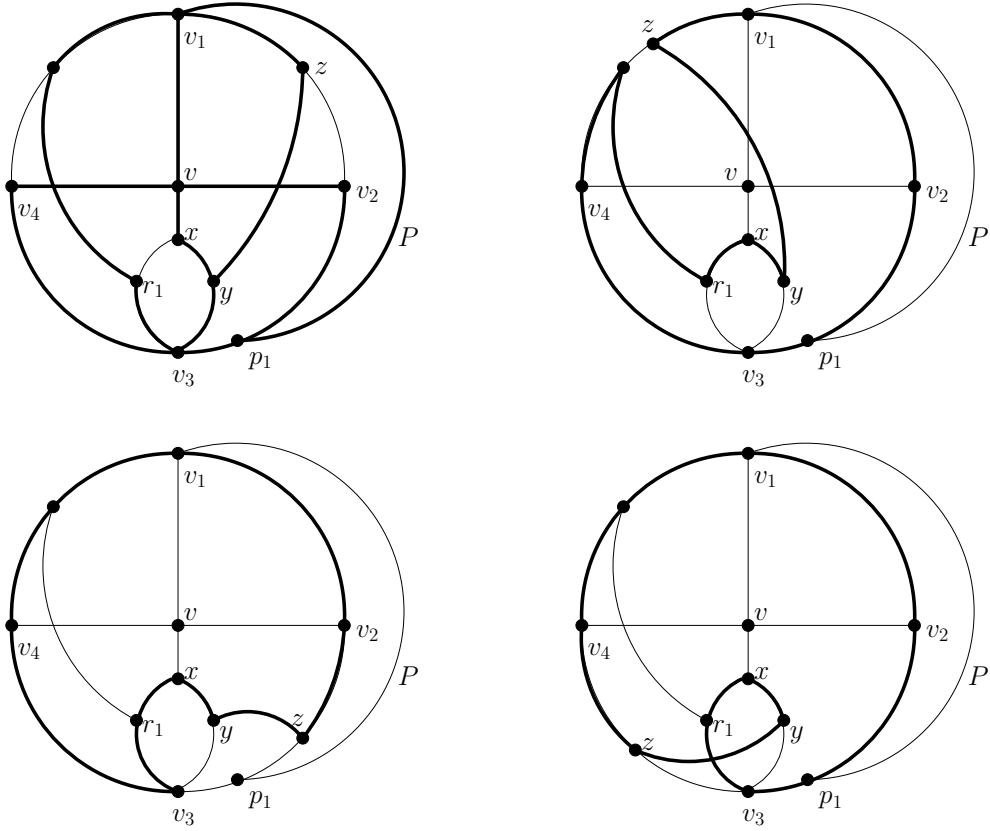}
\caption[Case (c)(ii)1. Claim \ref{claim_cii1_R}: $S'_2 \neq R'_2$.]{Case (c)(ii)1. Claim \ref{claim_cii1_R}: four possible configurations where $z$ lies on a rim segment of $H$. The first figure shows a $K_5^-$-subdivision in bold; the other three figures show the rim of a forbidden $W_4$-subdivision in bold.}
\label{case-cii-3a}
\end{center}
\end{figure}

Assume, then, that $y = v_3$. Then $P''_3$ is identical to $P'_3$, so there must be some other spoke of $H''$ which is shorter than its corresponding spoke in $H'$. In other words, at least one of $H''$'s spokes, $P''_i$, is a proper initial segment of some path  $P_i \in \{P_1, P_2, P_4\}$ that is a spoke of $H$ and of $H'$, so that the spoke-meets-rim vertex $v''_i$ of $P''_i$ lies internally on $P_i$. Let $R''_i$ be one of the segments of $H''$'s rim which has $v''_i$ as one of its endpoints. There is a subpath $S''_i$ of $R''_i$ from $v''_i$ to some vertex $w''_i$ in $H\cup P$, such that $S''_i$ does not meet $H\cup P$ internally.

Suppose firstly that $i \in \{2, 4\}$, that is, $v''_i$ lies internally on either $P_2$ or $P_4$. Regardless of where $w''_i$ lies in $H\cup P$, a $W_4$-subdivision can be formed which violates the shortness of $H$. The more complex of these cases are illustrated in Figure \ref{case-cii-4}.

\begin{figure}[H]
\begin{center}
\includegraphics[width=\textwidth]{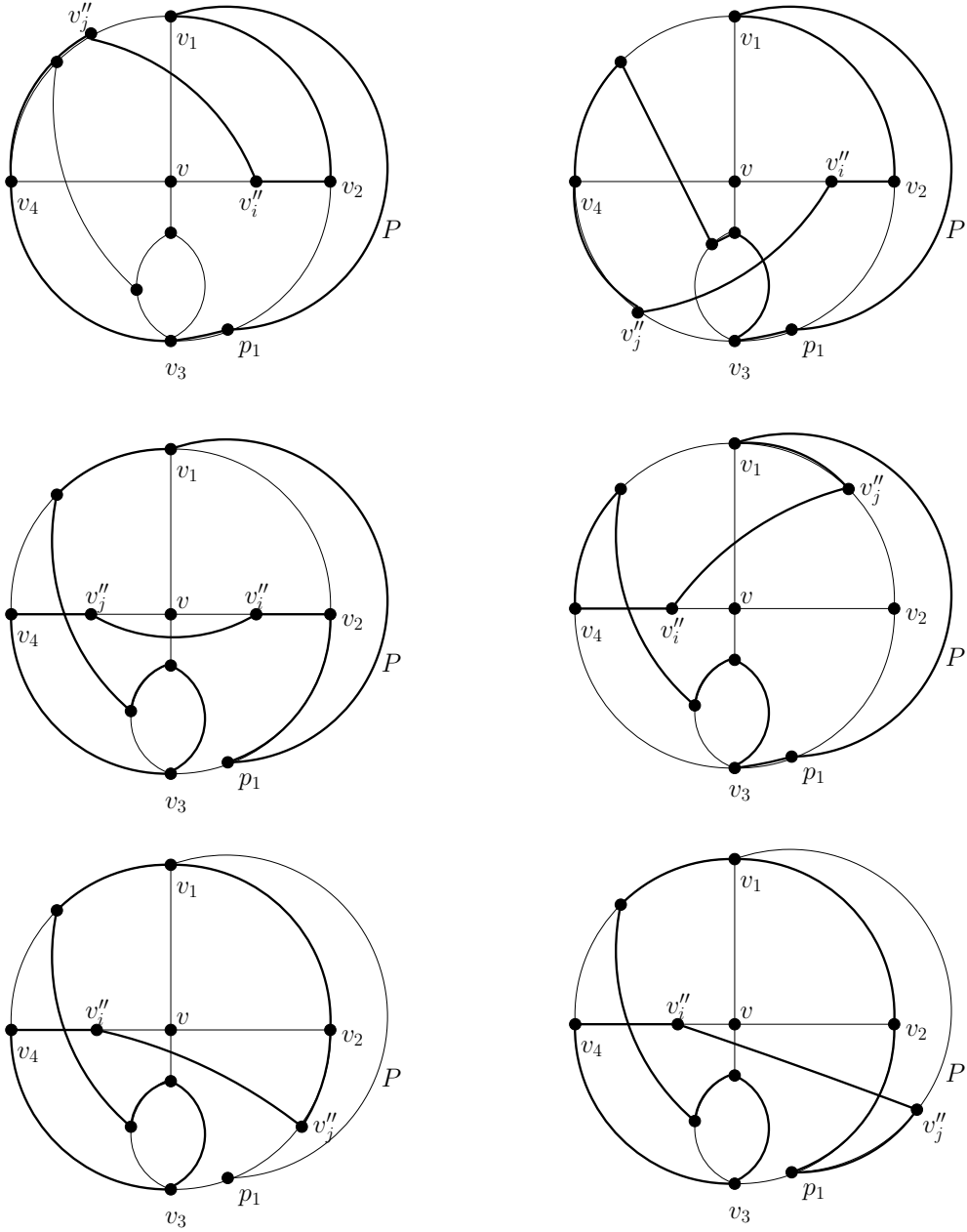}
\caption[Case (c)(ii)1. Claim \ref{claim_cii1_R}: $v''_i$ on $P_2$ or $P_4$]{Case (c)(ii)1. Claim \ref{claim_cii1_R}: $S''_i$ from $P_2 \cup P_4$ to some other part of $H\cup P$. The rim of the new short $W_4$-subdivision is shown in bold.}
\label{case-cii-4}
\end{center}
\end{figure}

Assume then that the spokes of $H''$ that lie along $P_2$ and $P_4$ are \emph{not} proper initial segments, that is, $H''$ has spoke-meets-rim vertices at $v_2$ and $v_4$, as well as at $v_3$. Thus, we can assume that $i = 1$, that is, $P''_1$ forms a proper initial segment of $P_1$, and $v''_i$ lies internally on $P_1$. The other three spokes of $H''$ are $P_2$, $P'_3$, and $P_4$.

Suppose $w''_i$ lies on $v_3P_3x$ internally. (This is possible, since there may be parts of this path that do not coincide with $P'_3$.) Then a $K_5^-$-subdivision can be formed in $G$, as shown in Figure \ref{case-cii-5}. (The figure shows $P_y$ not interacting with $v_3P_3x$, but the situation would be no different if these paths did interact, since $Q'$ is not used. Note also that $S''_i$ cannot meet $P_x$, since $P_x$ forms part of a spoke of $H''$, while $S''_i$ forms part of the rim.)

\begin{figure}[H]
\begin{center}
\includegraphics[width=0.5\textwidth]{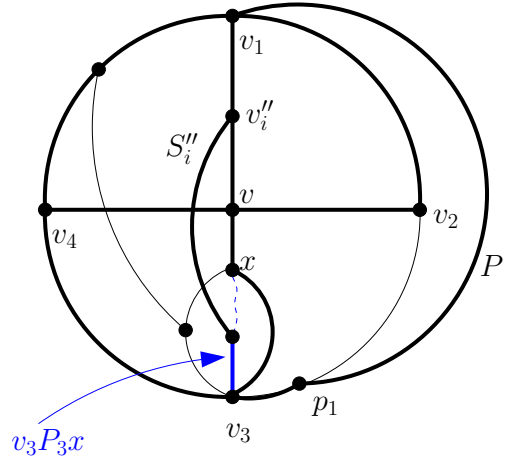}
\caption[Case (c)(ii)1. Claim \ref{claim_cii1_R}: $S''_i$ from $P_1$ to $P_3$.]{Case (c)(ii)1. Claim \ref{claim_cii1_R}: $S''_i$ from $P_1$ to $v_3P_3x$ creates a $K_5^-$-subdivision.}
\label{case-cii-5}
\end{center}
\end{figure}

Suppose then that $S''_i$ does not meet $v_3P_3x$ internally, and thus does not meet any of $H$'s spokes internally. Then there exists a $W_4$-subdivision whose spokes coincide with $H''$ except on $P'_3$, and has the spoke $P_3$ instead of $P'_3$, such that this $W_4$-subdivision violates the shortness of $H$.

Assume then that $H'$ is short. This means that $H'\cup P\cup Q'$ meets the requirements of the configuration addressed earlier in Case (c)(i) (where $U_3$ has vertices of attachment internally on $R_4$), where a $K_5^-$-subdivision is shown to exist in $G$.

This completes the proof of Claim \ref{claim_cii1_R}.
\renewcommand{\qedsymbol}{}
\end{proof}

Assume then that no such path $R$ exists.

If there exists a path from $(P_x\cup P_y) - \{v_3, x\}$ to $(p_1R_2v_3\cup R_3) - v_3$ that is internally disjoint from $(H\cup P) - P_3$, then a $W_4$-subdivision exists centred on $v$ such that $H$ is no longer short. Figure \ref{case-cii-endR} shows one possible configuration of this situation. Assume then that no such path exists. 

\begin{figure}[H]
\begin{center}
\includegraphics[width=0.5\textwidth]{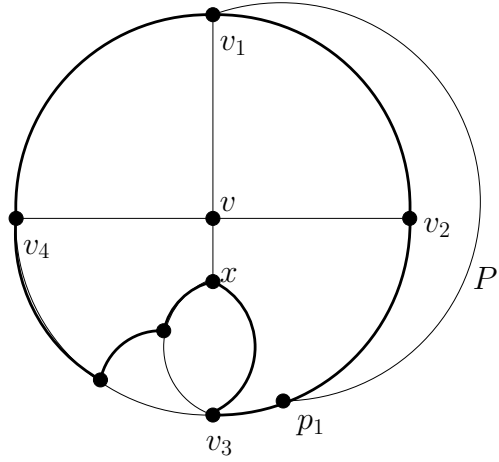}
\caption[Case (c)(ii)1. Path from $(P_x\cup P_y) - \{v_3, x\}$ to $(p_1R_2v_3\cup R_3) - v_3$.]{Case (c)(ii)1. Path from $(P_x\cup P_y) - \{v_3, x\}$ to $(p_1R_2v_3\cup R_3) - v_3$ creates a $W_4$-subdivision (rim shown in bold) which violates the shortness of $H$.}
\label{case-cii-endR}
\end{center}
\end{figure}

From the non-existence of path $R$ (Claim \ref{claim_cii1_R}) and the conclusion of the previous paragraph, we conclude that the removal of $v_3$ and $x$ will disconnect the graph, separating $(P_x\cup P_y)\setminus \{v_3, x\}$ from $(H\cup P) - P_3$.  Thus the 4-connectivity of $G$ is violated.

\textbf{2.} Assume then that $U_3$ has no vertices of attachment on $P_3 - v_3$, but only on $p_1R_2v_3$ or $R_3$.

Let $R'_3 = p_1R_2v_3\cup R_3$, that is, the section of the rim of $H$ from $p_1$ to $v_4$ that includes $v_3$.

Let $q_3$ be $U_3$'s vertex of attachment on $R'_3$ that lies closest to $v_4$ along $R'_3$, and let $Q_3$ be a path in $U_3$ from $v_3$ to $q_3$. Let $q_2$ be $U_3$'s vertex of attachment on $R'_3$ that lies closest to $p_1$ along $R'_3$, and let $Q_2$ be a path in $U_3$ from $v_3$ to $q_2$. Note that $q_3$ must lie on $R_3$ and $q_2$ must lie on $R_2$, since $v_3$ is already known to be a vertex of attachment of $U_3$ (although it is possible that either $q_3 = v_3$ or $q_2 = v_3$, but not both).

Let $\mathcal{R}_{a,b}$ be the set of all pairs of internally disjoint paths $\{R_a, R_b\}$, such that:

\begin{itemize}
\item $R_a$ and $R_b$ have shared endpoints, one of which is $v_3$, and the other some vertex on $v_3R_2p_1$;
\item neither $R_a$ nor $R_b$ meet $(H\cup P) - v_3R_2p_1$ internally (they may, however, interact with $v_3R_2p_1$); and
\item $\mathcal{R}_{a,b}$ excludes the pair of identical paths $\{v_3w, v_3w\}$, where $w$ is $v_3$'s neighbour on $R_2$.
\end{itemize}

This set includes the trivial pair of paths $\{v_3, v_3\}$, so is non-empty. Note that it also includes the pair of paths $\{Q_2, v_3R_2q_2\}$ (which may be the trivial pair of paths $\{v_3, v_3\}$, if $q_2 = v_3$).

Let $\{R_w, R_x\}$ be a pair of paths belonging to $\mathcal{R}_{a,b}$ with endpoints $v_3$ and $x$, chosen so that the distance between $p_1$ and $x$ along $v_3R_2p_1$ is minimised. Note that $x$ lies somewhere on $q_2R_2p_1$.

Suppose there exists a path $R$ from a vertex $r_1$ that lies in $(R_w\cup R_x) - \{v_3, x\}$, to a vertex $r_2$ in $G_2$, such that this path is internally disjoint from $(H - v_3R_2x)\cup P\cup R_w\cup R_x$. Without loss of generality, suppose $r_1$ lies internally on $R_w$.

Form a new $W_4$-subdivision $H'$ by replacing the part of $H$'s rim formed by $v_3R_2x$ with $R_x$. Note that since $H'$ has the same spokes as $H$, $H'$ must also be short.

Let $Q'$ be the path from $v_3$ to $G_2$ formed from $R$ and $v_3R_wr_1$. This path is internally disjoint from $H'\cup P$. Let $U'_3$ be the $H'\cup P$-bridge of $G$ containing $Q'$. If $U'_3$ has a vertex of attachment on $P_1 - v$, $R_1$, $v_2R_2p_1 - p_1$, or $P - p_1$, then a $K_{5}^{-}$-subdivision can be formed, as described at the beginning of Case (c). If $U'_3$ has a vertex of attachment on $P_2$ or $P_4$, then there exists a $W_4$-subdivision in $G$ centred on $v$ such that $H'$ is no longer short, which is a contradiction. So again, the only remaining place in $G_2$ where $U'_3$ can have a vertex of attachment is internally on $R_4$. The graph then meets the requirements of the configuration addressed earlier in Case (c)(i) (where $U_3$ has vertices of attachment internally on $R_4$).

Suppose then that no such path $R$ exists.

Let $\mathcal{R}_{c,d}$ be the set of all pairs of internally disjoint paths $\{R_c, R_d\}$, such that:

\begin{itemize}
\item $R_c$ and $R_d$ have shared endpoints, one of which is $v_3$, and the other some vertex on $R_3$; 
\item neither $R_c$ nor $R_d$ meet $(H\cup P) - R_3$ internally (they may, however, interact with $R_3$); and
\item $\mathcal{R}_{c,d}$ excludes the pair of identical paths $\{v_3w, v_3w\}$, where $w$ is $v_3$'s neighbour on $R_3$.
\end{itemize}

This set includes the trivial pair of paths $\{v_3, v_3\}$, so is non-empty. Note that it also includes the pair of paths $\{Q_3, v_3R_3q_3\}$ (which may be the trivial pair of paths $\{v_3, v_3\}$, if $q_3 = v_3$).

Let $\{R_y, R_z\}$ be a pair of paths belonging to $\mathcal{R}_{c,d}$ with endpoints $v_3$ and $y$, chosen so that the distance between $v_4$ and $y$ along $R_3$ is minimised.

Suppose now that there exists a path $R$ from a vertex $r_1$ that lies in $(R_y\cup R_z) - \{v_3, y\}$, to a vertex $r_2$ in $G_2$, such that this path is internally disjoint from $(H - v_3R_3y)\cup P\cup R_y\cup R_z$. Without loss of generality, suppose $r_1$ lies internally on $R_z$.

Form a new $W_4$-subdivision $H'$ by replacing the part of $H$'s rim formed by $v_3R_3y$ with $R_y$. Note that since $H'$ has the same spokes as $H$, $H'$ must also be short.

Let $Q'$ be the path from $v_3$ to $G_2$ formed from $R$ and $v_3R_zr_1$. This path is internally disjoint from $H'\cup P$. Let $U'_3$ be the $H'\cup P$-bridge of $G$ containing $Q'$. If $U'_3$ has a vertex of attachment on $P_1 - v$, $R_1$, $v_2R_2p_1 - p_1$, or $P - p_1$, then a $K_{5}^{-}$-subdivision can be formed, as described at the beginning of Case (c). If $U'_3$ has a vertex of attachment on $P_2$ or $P_4$, then there exists a $W_4$-subdivision in $G$ centred on $v$ such that $H'$ is no longer short, which is a contradiction. So again, the only remaining place in $G_2$ where $U'_3$ can have a vertex of attachment is internally on $R_4$. The graph then meets the requirements of the configuration addressed earlier in Case (c)(i) (where $U_3$ has vertices of attachment internally on $R_4$).

Suppose then that no such path $R$ exists.

If there exists a path from an internal vertex of $P_3$ to $R'_3 - v_3$ that is internally disjoint from $H\cup P$, then a $W_4$-subdivision exists such that $H$ is no longer short. Assume then that no such path exists.

Then the removal of $v_3$, $x$, and $y$ (at least two of which must be distinct vertices) will disconnect the graph, separating $(R_w\cup R_x)\setminus \{v_3, x\}$ and $(R_y\cup R_z)\setminus \{v_3, y\}$ (at least one of which must be non-empty) from $(H\cup P) - R'_3$. This contradicts the 4-connectivity of $G$.


\vspace{0.1in}
\noindent \textbf{Case (d): $p_1$ is an internal vertex of $P_3$.}

Suppose $p_1$ is an internal vertex of $P_3$. Without loss of generality, assume that $P$ is chosen to minimise the distance between $p_1$ and $v_3$ along $P_3$.

By the 4-connectivity of $G$, there must exist some fourth neighbour $u_3$ of $v_3$, such that $u_3 \notin N_{H}(v_3)$. Let $U_3$ be the $(H\cup P)$-bridge of $G$ containing the edge $v_3u_3$.

If $U_3$ has some vertex of attachment $u$, such that $u \in \{v, v_1\}$, or $u$ is an internal vertex of $P$ or $vP_{3}p_{1}$, then there exists some path $Q$ from $v_3$ to $u$ that is internally disjoint from $H\cup P$. A $K_{5}^{-}$-subdivision exists in $G$ in each of these cases. Assume then that $U_3$ has no such vertex of attachment $u$.

If $U_3$ has some vertex of attachment $u$ such that $u$ is an internal vertex of $R_1$ or $R_4$, then $G$ contains a graph that is symmetrically equivalent to that of the previous case (Case (c)). Assume then that $U_3$ has no such vertex of attachment $u$.

If $U_3$ has any vertices of attachment that are internal vertices of $P_2$ or $P_4$, then there exists some path from $v_3$ to an internal vertex of $P_2$ or $P_4$, such that this path is internally disjoint from $H$. This results in a new $W_4$-subdivision $H'$ centred on $v$, such that $H$ is no longer short. Assume then that $U_3$ contains no vertices of attachment that are internal vertices of $P_2$ or $P_4$.

It can be assumed, then, that $U_3$'s vertices of attachment lie only on the following paths:

\begin{itemize}
\item[(i)] $P_1$ (internally)
\item[(ii)] $R_2$, $R_3$, or $p_{1}P_{3}v_{3}$ (potentially at their endpoints)
\end{itemize}

These cases are considered below.

\vspace{0.1in}
\noindent \textbf{(i) $U_3$ has some internal vertex of $P_1$ as a vertex of attachment}

Suppose $U_3$ has some vertex of attachment $q_3$ that is an internal vertex of $P_1$. Thus, there exists some path $Q$ from $v_3$ to $q_3$ that meets $H\cup P$ only at its endpoints. See Figure \ref{case-di}.

\begin{figure}[H]
\begin{center}
\includegraphics[width=0.5\textwidth]{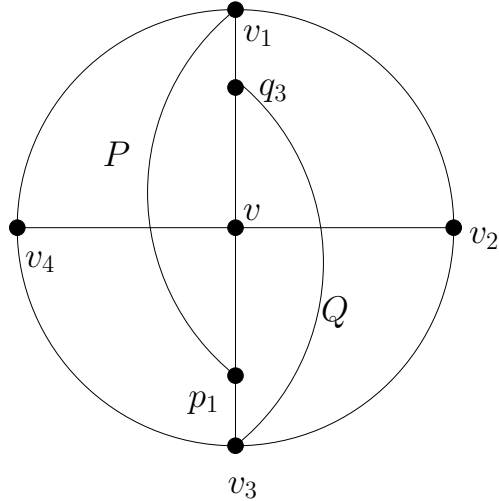}
\caption[Case (d)(i): Paths $P$ and $Q$.]{Case (d)(i). Path $Q$ from $v_3$ to $q_3$, where $q_3$ is an internal vertex of $P_1$.}
\label{case-di}
\end{center}
\end{figure}

Let $G_1 = R_1\cup R_2\cup P_2$. Let $G_2 = R_3\cup R_4\cup P_4$. By 4-connectivity of $G$, there must exist some path $R$ from $G_1$ to $G_2$ that is disjoint from $\{v_1, v, v_3\}$. If this path meets $P$, $Q$, $P_1$, or $P_3$, then checking each of the resulting graphs (which fall into nine isomorphism classes up to symmetry) shows that in all cases, a $W_4$-subdivision exists centred on $v$ such that $H$ is no longer short. An example of one such graph is shown in Figure \ref{case-di-short}.

\begin{figure}[H]
\begin{center}
\includegraphics[width=0.5\textwidth]{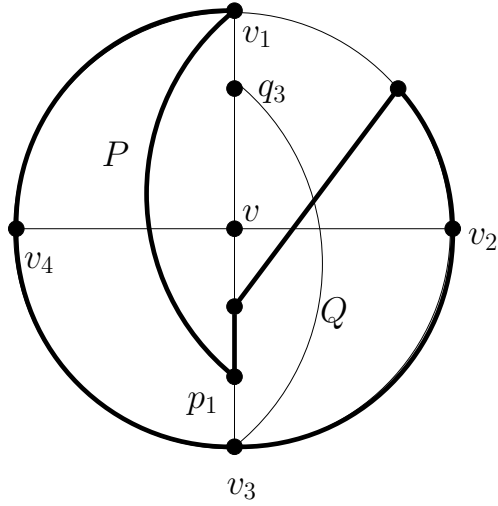}
\caption[Case (d)(i): $W_4$-subdivision with shorter spokes than $H$.]{Case (d)(i). Any path from $(G_1\cup G_2) \setminus \{v_1, v, v_3\}$ to $(P\cup Q\cup P_1\cup P_3)\setminus \{v_1, v, v_3\}$ results in a $W_4$-subdivision with shorter spokes than $H$. This example meets $P_3$. The rim of the new $W_4$-subdivision is shown in bold.}
\label{case-di-short}
\end{center}
\end{figure}

Assume then that $R$ does not meet $P$, $Q$, $P_1$, or $P_3$.

There are nine possible placements of the endpoints of $R$:

\begin{tabular}{l l l}
1. $R_1$ and $R_3$ & 2. $R_1$ and $R_4$ & 3. $R_1$ and $P_4$ \\
4. $R_2$ and $R_3$ & 5. $R_2$ and $R_4$ & 6. $R_2$ and $P_4$ \\
7. $P_2$ and $R_3$ & 8. $P_2$ and $R_4$ & 9. $P_2$ and $P_4$ \\
\end{tabular}

Table \ref{case-di-table} shows how each case number listed is assigned to each possible pairing of endpoints.

\begin{table}[H]
\begin{center}
\begin{tabular}{|p{0.6cm}||c|c|c|} \hline
\hfill $G_2$\newline $G_1$ & $R_3$ & $R_4$ & $P_4$ \\ \hline\hline
\textbf{$R_1$} & 1 & 2 & 3 \\ \hline
\textbf{$R_2$} & 4 & 5 & 6 \\ \hline
\textbf{$P_2$} & 7 & 8 & 9 \\ \hline
\end{tabular}
\caption[Case (d)(i). Table showing possible endpoints of $R$.]{Case (d)(i). Table showing assigned case numbers of all possible endpoints of $R$.}
\label{case-di-table}
\end{center}
\end{table}

In cases 1 and 5, a $K_5^-$-subdivision exists in $G$. In cases 3, 6, 7, 8, and 9, a $W_4$-subdivision exists such that $H$ is no longer short. Cases 2 and 4 remain. In both these cases we may assume that neither endpoint of $R$ is $v_2$ or $v_4$, since such cases fall within other cases too.

Let $r_1$ be the closest vertex to $v_2$ along $R_1$ that forms an endpoint of such a path $R$ as described in case 2 above, if such a path $R$ exists.

Let $r_2$ be the closest vertex to $v_2$ along $R_2$ that forms an endpoint of such a path $R$ as described in case 4 above, if such a path $R$ exists.

Then the graph can be disconnected by the removal of $r_1$ (or $v_1$, if $r_1$ does not exist), $v$, and $r_2$ (or $v_3$, if $r_2$ does not exist), placing $v_2$ in a separate component from $G_2$ (noting that $v_2\notin\{r_1,r_2\}$), thus contradicting the 4-connectivity of $G$.

\vspace{0.1in}
\noindent \textbf{(ii) $U_3$'s vertices of attachment all lie on $R_2$, $R_3$, or $p_{1}P_{3}v_{3}$}

Suppose $U_3$ only has vertices of attachment on $R_2$, $R_3$, or $p_{1}P_{3}v_{3}$.

Let $G_1 = R_2\cup R_3\cup p_1P_3v_3$. Let $G_2 = (H\cup P) - G_1$. Note that throughout this section of the proof, $U_3$'s vertices of attachment are all contained in $G_1$.

\textbf{1.} Suppose $U_3$ has some vertex of attachment on $p_1P_3v_3$ other than $v_3$.

Let $q_3$ be $U_3$'s vertex of attachment that lies closest to $p_1$ along $P_3$.

Suppose there exists some path $R$ from some vertex $r_1$ on $v_3P_3q_3 - q_3$ to some vertex $r_2$ in $G_2$, such that $R$ is internally disjoint from $H\cup P$. If $r_2$ lies on $P - p_1$ or $vP_3p_1 - p_1$, then a $K_5^-$-subdivision exists in $G$. If $r_2$ lies internally on $P_2$, $P_4$, $R_1$, or $R_4$, then a $W_4$-subdivision exists centred on $v$ such that $H$ is no longer short. Assume, then, that $r_2$ lies internally on $P_1$.

Let $Q'$ be the path formed from $R$ and $v_3P_3r_1$. Let $P'_3$ be a path formed from $vP_3q_3$ and some path $P_x$ in $U_3$ from $q_3$ to $v_3$, such that $P'_3$ and $Q'$ are disjoint except at $v_3$. Let $H'$ be a $W_4$-subdivision that coincides with $H$ except on $P_3$, and has the spoke $P'_3$ instead of $P_3$.

Suppose $H'$ is not short. Then there exists another $W_4$-subdivision, $H''$, centred on $v$, whose spokes are initial segments of the spokes of $H'$, with at least one of these initial segments being proper. Let $P''_3$ be the spoke of $H''$ that is an initial segment $vP'_3y$ of $P'_3$, where $y$ is the spoke-meets-rim vertex of $H''$ that lies on $P'_3$. There must exist two internally-disjoint paths $R'_2$ and $R'_3$ from $y$ to two of the paths $P_1$, $P_2$ and $P_4$, that form two segments of the rim of $H''$. Thus, at least one of these rim segments --- assume $R'_2$ without loss of generality --- meets either $P_2$ or $P_4$. Recall $P'_3$ is composed of two subpaths, $vP_3q_3$ and $P_x$. If $y$ lies on $vP_3q_3$, then $H''$ violates the shortness of $H$. Thus, $y$ must lie on the path $P_x - q_3$. But $P_x$ is contained in the $(H\cup P)$-bridge $U_3$, which only has vertices of attachment on $R_2$, $R_3$, or $p_{1}P_{3}v_{3}$. Since $R'_2$ meets either $P_2$ or $P_4$, it cannot meet $P_x$ internally, or $U_3$ would also contain vertices of attachment on either $P_2$ or $P_4$.

Assume then that $y = v_3$. Then $P''_3$ is identical to $P'_3$ so there must be some other spoke $P''_i$ of $H''$ which is a  \emph{proper} initial segment of another spoke $P_i$ of $H'$, where $i \in \{1, 2, 4\}$, so that some spoke-meets-rim vertex $v''_i$ of $P''_i$ lies internally on $P_i$. Let $R''_i$ be one of the segments of $H''$'s rim which has $v''_i$ as one of its endpoints. There is a subpath $S''_i$ of $R''_i$ from $v''_i$ to some vertex $w''_i$ in $H\cup P$, such that $S''_i$ does not meet $H\cup P$ internally. Regardless of where $v''_i$ and $w''_i$ lie, the configuration results in either the existence of a $K_5^-$-subdivision, or a $W_4$-subdivision which violates the shortness of $H$.

Assume, then, that $H'$ is short. Then we can use similar arguments to those given in Case (d)(i) above, but replacing $H$ and $Q$ with $H'$ and $Q'$ respectively.

Assume then that no such path $R$ exists.

If there exists a path from $v_3P_3q_3 - v_3$ to $(R_2\cup R_3) - v_3$ that is internally disjoint from $H\cup P$, then a $W_4$-subdivision exists such that $H$ is no longer short. Assume then that no such path exists. (Thus, $U_3$'s vertices of attachment must \emph{all} lie on $v_3P_3q_3$, since if some were on $(R_2\cup R_3) - v_3$, then there \emph{would} exist such a path, internally contained in $U_3$.)

Let $\mathcal{U}$ be the set of all $p_1P_3v_3$-bridges of $G$ except the one containing $R_2\cup R_3\cup G_2$. (Note that $U_3$ is one such bridge in $\mathcal{U}$.) Let $A$ be the set of all vertices of attachment of bridges in $\mathcal{U}$, and let $q'_3$ be the vertex in $A$ closest to $p_1$ along $P_3$.

The same arguments used for $U_3$ above can be used to show that there is no path from $v_3P_3q'_3 - q'_3$ to $G_2$ that is internally disjoint from $H\cup P$, and there is no path from $v_3P_3q'_3 - v_3$ to $(R_2\cup R_3) - v_3$ that is internally disjoint from $H\cup P$.

There cannot be a path from $v_3P_3q'_3 - q'_3$ to $q'_3P_3p_1$ that is internally disjoint from $p_1P_3v_3$, as this path would belong to some bridge in $\mathcal{U}$, which would contradict the choice of $q'_3$.

Thus, the removal of $v_3$ and $q'_3$ will disconnect the graph, placing the remaining vertices in $\mathcal{U}$ in a different component from the rest of $G$. Note that there must exist at least one such remaining vertex: either $u_3$ (in $U_3$) is a distinct vertex from $q'_3$, or $v_3P_3q'_3$ contains an internal vertex (otherwise there would be a double edge from $v_3$ to $q'_3$).

\textbf{2.} Assume then that $U_3$ has no vertices of attachment on $v_3P_3p_1 - v_3$, but only on $R_2$ or $R_3$.

Let $R_{2,3} = R_2\cup R_3$.

Let $q_2$ be $U_3$'s vertex of attachment that lies closest to $v_2$ along $R_{2,3}$, and let $Q_2$ be a path in $U_3$ from $v_3$ to $q_2$. Let $q_3$ be $U_3$'s vertex of attachment that lies closest to $v_4$ along $R_{2,3}$, and let $Q_3$ be a path in $U_3$ from $v_3$ to $q_3$. Note that, since $v_3$ is a vertex of attachment of $U_3$, it is possible that either $q_2 = v_3$ or $q_3 = v_3$ (but not both), and furthermore, we know that $q_2$ must lie somewhere on $R_2$, and $q_3$ must lie somewhere on $R_3$.

Let $\mathcal{R}_{a,b}$ be the set of all pairs of internally disjoint paths $\{R_a, R_b\}$ such that:

\begin{itemize}
\item $R_a, R_b$ have shared endpoints, one of which is $v_3$, and the other some internal vertex of $R_2$;
\item neither $R_a$ nor $R_b$ meet $(H\cup P) - R_2$ internally, however, they may interact with $R_2$; and
\item $\mathcal{R}_{a,b}$ excludes the pair of identical paths $\{v_3w, v_3w\}$, where $w$ is $v_3$'s neighbour on $R_2$.
\end{itemize}

This set includes the trivial pair of paths $\{v_3, v_3\}$, so is non-empty. Note that it also includes the pair of paths $\{v_3R_2q_2, Q_2\}$ (which may be the trivial pair of paths $\{v_3, v_3\}$, if $q_2 = v_3$).

Let $\{R_w, R_x\}$ be a pair of paths belonging to $\mathcal{R}_{a,b}$ with endpoints $v_3$ and $x$, chosen so that the distance between $v_2$ and $x$ along $R_2$ is minimised. Note that $x$ lies somewhere on $q_2R_2v_2$.

Suppose there exists a path $R$ from a vertex $r_1$ that lies on $(R_w\cup R_x) - \{v_3, x\}$ to a vertex $r_2$ in $G_2$, such that $R$ is internally disjoint from $(H - v_3R_2x)\cup P\cup R_w\cup R_x$. Without loss of generality, suppose $r_1$ lies internally on $R_w$.

Form a new $W_4$-subdivision, $H'$, by replacing the part of $H$'s rim formed by $v_3R_2x$ with $R_x$. Note that since $H'$ has the same spokes as $H$, $H'$ is also short.

Let $Q'$ be the path from $v_3$ to $G_2$ formed from $R$ and $v_3R_wr_1$. This path is internally disjoint from $H'\cup P$. Then $H' \cup P \cup Q'$ meets the requirements of one of the configurations addressed in either the first few paragraphs of Case (d), or in Case (d)(i). A $K_5^-$-subdivision has already been shown to exist under such conditions.

Assume then that such a path $R$ does not exist.

Now, define $\mathcal{R}_{c,d}$ as the set of all pairs of internally disjoint paths $\{R_c, R_d\}$ such that:

\begin{itemize}
\item $R_c, R_d$ have shared endpoints, one of which is $v_3$, and the other some internal vertex of $R_3$;
\item neither $R_c$ nor $R_d$ meet $(H\cup P) - R_3$ internally, however, they may interact with $R_3$; and
\item $\mathcal{R}_{c,d}$ excludes the pair of identical paths $\{v_3w, v_3w\}$, where $w$ is $v_3$'s neighbour on $R_3$.
\end{itemize}

This set includes the trivial pair of paths $\{v_3, v_3\}$, so is non-empty. Note that it also includes the pair of paths $\{v_3R_3q_3, Q_3\}$ (which may be the trivial pair of paths $\{v_3, v_3\}$, if $q_3 = v_3$).

Let $\{R_y, R_z\}$ be a pair of paths belonging to $\mathcal{R}_{c,d}$ with endpoints $v_3$ and $y$, chosen so that the distance between $v_4$ and $y$ along $R_3$ is minimised. Note that $y$ lies somewhere on $q_3R_3v_4$.

Suppose there exists a path $R$ from a vertex $r_1$ that lies on $(R_y\cup R_z) - \{v_3, y\}$ to a vertex $r_2$ in $G_2$, such that $R$ is internally disjoint from $(H - v_3R_3y)\cup P\cup R_y\cup R_z$. Without loss of generality, suppose $r_1$ lies internally on $R_y$.

Form a new $W_4$-subdivision, $H'$, by replacing the part of $H$'s rim formed by $v_3R_3y$ with $R_z$. Note that since $H'$ has the same spokes as $H$, $H'$ is also short.

Let $Q'$ be the path from $v_3$ to $G_2$ formed from $R$ and $v_3R_yr_1$. This path is internally disjoint from $H'\cup P$. Then $H' \cup P \cup Q'$ meets the requirements of one of the configurations addressed in either the first few paragraphs of Case (d), or in Case (d)(i). A $K_5^-$-subdivision has already been shown to exist under such conditions.

Assume then that such a path $R$ does not exist.

If there exists a path from $p_1P_3v_3 - v_3$ to $R_{2,3} - v_3$ that is internally disjoint from $H\cup P$, then a $W_4$-subdivision exists such that $H$ is no longer short. Assume then that no such path exists.

Then the removal of $v_3$, $x$, and $y$ (at least two of which must be distinct vertices) will disconnect the graph, separating $(R_w\cup R_x) \setminus \{v_3,x\}$ and $(R_y\cup R_z) \setminus \{v_3,y\}$ (at least one of which must be non-empty) from $(H\cup P) - R_{2,3}$. This contradicts the 4-connectivity of $G$.


\vspace{0.1in}
\noindent \textbf{Case (e): $p_1$ lies on $R_1$, $R_4$, or $P_1$.}

Let $U_1$ be the $H$-bridge of $G$ containing $P$. Let $G_1 = R_1\cup R_4\cup P_1$. Let $G_2 = H - G_1$. Note that $U_1$'s vertices of attachment are all contained in $G_1$.

\textbf{1.} Suppose $U_1$ has at least one vertex of attachment on $P_1$ other than $v_1$.

Let $q_1$ be $U_1$'s vertex of attachment that lies closest to $v$ along $P_1$. Let $P_x$ be some path contained in $U_1$ from $v_1$ to $q_1$ that does not meet $H$ internally.

Suppose there exists some path $R$, with endpoint $r_1$ on $v_1P_1q_1 - q_1$, and endpoint $r_2$ in $G_2$, such that $R$ is internally disjoint from $H$. (Note that $R$ must also be internally disjoint from the vertices in $U_1 - H$, since otherwise $R$ would be contained in $U_1$, and we already know that $U_1$ has no vertices of attachment in $G_2$.) Let $P'$ be the path from $v_1$ to $r_2$ formed from $v_1P_1r_1$ and $R$. Let $P'_1$ be a path from $v_1$ to $v$, formed from $P_x$ and $q_1P_1v$. Note that $P'$ and $P'_1$ are disjoint except at $v_1$, and do not meet $H - P_1$ internally.

Let $H'$ be a $W_4$-subdivision that coincides with $H$ except on $P_1$, and has the spoke $P'_1$ instead of $P_1$. Suppose $H'$ is not short. Then there exists another $W_4$-subdivision, $H''$, centred on $v$, with spokes that are initial segments of the spokes of $H'$, with at least one of these initial segments being proper. Let $P''_1$ be the spoke of $H''$ that is an initial segment $vP'_1y$ of $P'_1$, where $y$ is the spoke-meets-rim vertex of $H''$ that lies on $P'_1$. There must exist two internally-disjoint paths $R'_1$ and $R'_4$ from $y$ to two of the three paths $P_2$, $P_3$, and $P_4$, forming two segments of the rim of $H''$. Thus, at least one of these paths --- assume $R'_1$, without loss of generality --- goes from $y$ to either $P_2$ or $P_4$. Recall $P'_1$ is composed of two subpaths, $vP_1q_1$ and $P_x$. If $y$ lies on $vP_1q_1$, then $H''$ violates the shortness of $H$. Thus, $y$ must lie on the path $P_x$. But $P_x$ is contained in the $H$-bridge $U_1$, whose vertices of attachment are all contained in $G_1$. Since $R'_1$ meets $G_2$ (on either $P_2$ or $P_4$, internally), it cannot meet $P_x$ internally, or $U_1$ would also contain vertices of attachment in $G_2$.

Assume then that $y = v_1$. Then $P''_1$ is identical to $P'_1$, so there must be some other spoke $P''_i$ of $H''$ which is a  \emph{proper} initial spoke of another spoke $P_i$ of $H'$, where $i \in \{2, 3, 4\}$, so that a spoke-meets-rim vertex $v''_i$ of $P''_i$ lies internally on $P_i$. Let $R''_i$ be one of the segments of $H''$'s rim which has $v''_i$ as one of its endpoints. There is a subpath $S''_i$ of $R''_i$ from $v''_i$ to some vertex $w''_i$ in $H\cup P$, such that $S''_i$ does not meet $H\cup P$ internally. Regardless of where $v''_i$ and $w''_i$ lie, either the configuration results in a $W_4$-subdivision which violates the shortness of $H$, or in a graph that has already been considered in a previous case.

We can assume, then, that $H'$ is short, and so can use similar arguments to those already given in Cases (a)--(d), but using $H'$ and $P'$ in place of $H$ and $P$.

Assume then that no such path $R$ exists.

If there exists a path from $v_1P_1q_1 - v_1$ to $(R_1\cup R_4) - v_1$ that is internally disjoint from $H$, then a $W_4$-subdivision exists such that $H$ is no longer short. Assume then that no such path exists. Thus, $U_1$'s vertices of attachment all lie on $P_1$.

Let $\mathcal{U}$ be the set of all $P_1$-bridges of $G$ except the one containing $R_1\cup R_4\cup G_2$. (Note that we have now shown $U_1$ to be one such bridge in $\mathcal{U}$.) Let $A$ be the set of all vertices of attachment of bridges in $\mathcal{U}$, and let $q'_1$ be the vertex in $A$ closest to $v$ along $P_1$.

The same arguments used for $U_1$ above (earlier in this subcase) can be used to show that there is no path from $v_1P_1q'_1 - q'_1$ to $G_2$ that is internally disjoint from $H$, and there is no path from $v_1P_1q'_1 - v_1$ to $(R_1\cup R_4) - v_1$ that is internally disjoint from $H$.

There cannot be a path from $v_1P_1q'_1 - q'_1$ to $q'_1P_1v$ that is internally disjoint from $P_1$, as this path would belong to some bridge in $\mathcal{U}$, which would contradict the choice of $q'_1$.

Thus, the removal of $v_1$ and $q'_1$ will disconnect the graph, placing the vertices in $\mathcal{U}$ in a different component from the rest of $G$. Note that there must exist at least one such remaining vertex: either $u_1$ (in $U_1$) is distinct from $q'_1$, or $v_1P_1q'_1$ contains an internal vertex (otherwise there would be a double edge from $v_1$ to $q'_1$).

\textbf{2.} Assume then that $U_1$ has no vertices of attachment on $P_1 - v_1$, but only on $R_1$ or $R_4$.

Let $R_{1,4} = R_1\cup R_4$.

Let $q_1$ be $U_1$'s vertex of attachment closest to $v_2$ along $R_{1,4}$, and let $Q_1$ be a path in $U_1$ from $v_1$ to $q_1$. Let $q_4$ be $U_1$'s vertex of attachment closest to $v_4$ along $R_{1,4}$, and let $Q_4$ be a path in $U_1$ from $v_1$ to $q_4$. Note that, since $v_1$ is a vertex of attachment of $U_1$, it is possible that either $q_1 = v_1$ or $q_4 = v_1$ (but not both), and furthermore, we know that $q_1$ must lie somewhere on $R_1$, and $q_4$ must lie somewhere on $R_4$.

Let $\mathcal{R}_{a,b}$ be the set of all pairs of internally disjoint paths $\{R_a, R_b\}$ such that:

\begin{itemize}
\item $R_a, R_b$ have shared endpoints, one of which is $v_1$, and the other some internal vertex of $R_1$;
\item neither $R_a$ nor $R_b$ meet $H - R_1$ internally, however, they may interact with $R_1$; and
\item $\mathcal{R}_{a,b}$ excludes the pair of identical paths $\{v_1w, v_1w\}$, where $w$ is $v_1$'s neighbour on $R_1$.
\end{itemize}

This set includes the trivial pair of paths $\{v_1, v_1\}$, so is non-empty. Note that it also includes the pair of paths $\{v_1R_1q_1, Q_1\}$ (which may be the trivial pair of paths $\{v_1 v_1\}$, if $q_1 = v_1$).

Let $\{R_w, R_x\}$ be a pair of paths belonging to $\mathcal{R}_{a,b}$ with endpoints $v_1$ and $x$, chosen so that the distance between $v_2$ and $x$ along $R_1$ is minimised. Note that $x$ lies somewhere on $q_1R_1v_2$.

Suppose there exists a path $R$ from a vertex $r_1$ that lies on $(R_w\cup R_x) - \{v_1, x\}$ to a vertex $r_2$ in $G_2$, such that $R$ is internally disjoint from $(H - v_1R_1x)\cup R_w\cup R_x$. Without loss of generality, suppose $r_1$ lies internally on $R_w$.

Form a new $W_4$-subdivision, $H'$, by replacing the part of $H$'s rim formed by $v_1R_1x$ with $R_x$. Note that since $H'$ has the same spokes as $H$, $H'$ is also short.

Let $P'$ be the path from $v_1$ to $G_2$ formed from $R$ and $v_1R_wr_1$. This path is internally disjoint from $H'$. Then $H' \cup P'$ meets the requirements of one of the configurations addressed in one of Cases (a), (b), (c), or (d), where a $K_5^-$-subdivision has already been shown to exist.

Assume then that such a path $R$ does not exist.

Now, define $\mathcal{R}_{c,d}$ as the set of all pairs of internally disjoint paths $\{R_c, R_d\}$ such that:

\begin{itemize}
\item $R_c, R_d$ have shared endpoints, one of which is $v_1$, and the other some internal vertex of $R_4$;
\item neither $R_c$ nor $R_d$ meet $H - R_4$ internally, however, they may interact with $R_4$; and
\item $\mathcal{R}_{c,d}$ excludes the pair of identical paths $\{v_1w, v_1w\}$, where $w$ is $v_1$'s neighbour on $R_4$.
\end{itemize}

This set includes the trivial pair of paths $\{v_1, v_1\}$, so is non-empty. It also includes the pair of paths $\{v_1R_4q_4, Q_4\}$ (which may be the trivial pair of paths $\{v_1, v_1\}$, if $q_1 = v_1$).

Let $\{R_y, R_z\}$ be a pair of paths belonging to $\mathcal{R}_{c,d}$ with endpoints $v_1$ and $y$, chosen so that the distance between $v_4$ and $y$ along $R_4$ is minimised. Note that $y$ lies somewhere on $q_4R_4v_4$.

Suppose there exists a path $R$ from a vertex $r_1$ that lies on $(R_y\cup R_z) - \{v_1, y\}$ to a vertex $r_2$ in $G_2$, such that $R$ is internally disjoint from $(H - v_3R_3y)\cup R_y\cup R_z$. Without loss of generality, suppose $r_1$ lies internally on $R_y$.

Form a new $W_4$-subdivision, $H'$, by replacing the part of $H$'s rim formed by $v_1R_4y$ with $R_z$. Since $H'$ has the same spokes as $H$, $H'$ is also short.

Let $P'$ be the path from $v_1$ to $G_2$ formed from $R$ and $v_1R_yr_1$. This path is internally disjoint from $H'$. Then $H' \cup P'$ meets the requirements of one of the configurations addressed in one of Cases (a), (b), (c), or (d), where a $K_5^-$-subdivision has already been shown to exist.

Assume then that such a path $R$ does not exist.

If there exists a path from $P_1 - v_1$ to $R_{1,4} - v_1$ that is internally disjoint from $H$, then a $W_4$-subdivision exists such that $H$ is no longer short. Assume then that no such path exists.

Then the removal of $v_1$, $x$, and $y$ (at least two of which must be distinct vertices) will disconnect the graph, separating $(R_w\cup R_x)\setminus \{v_1, x\}$ and $(R_y\cup R_z)\setminus \{v_1, y\}$ (at least one of which is non-empty) from $H - R_{1,4}$. This contradicts the 4-connectivity of $G$.

\end{proof}

\end{document}